\documentclass{article}

\usepackage{tikz}
\usetikzlibrary{calc}
\usepackage{cite}
\usepackage{graphicx}
\usepackage{algorithm}
\usepackage{algorithmic}
\usepackage{epsfig}
\usepackage{epstopdf}
\usepackage{amssymb,amsmath}
\usepackage{amsmath}
\usepackage{amsthm}
\usepackage{amsfonts}
\usepackage{subfigure}
\usepackage{psfrag}
\usepackage{rotating}
\usepackage{latexsym}
\usepackage{dsfont}
\usepackage{stmaryrd}
\usepackage{amsthm}
\usepackage{amssymb}
\usepackage{amsmath}

\usepackage[force,almostfull]{textcomp}
\usepackage{fancyvrb}
\usepackage{textcomp}
\usepackage{url}
\usepackage{ifdraft}
\usepackage{url}
\usepackage{multirow}
\usepackage{rotating}
\usepackage{xspace}
\usepackage{array}
\usepackage{multido}
\usepackage{flushend}
 \usepackage{enumitem}
\usepackage{enumerate}
\usepackage[latin1]{inputenc}
\newlength\figureheight 
\newlength\figurewidth 
\usepackage{graphics}
\usepackage{pgfplots}
\pgfplotsset{compat=newest}
\pgfplotsset{plot coordinates/math parser=false}

\usepackage{scalefnt}

\newtheoremstyle{specialcasestyle}{1mm}{1mm}{\upshape}{}{\bfseries\upshape}{.}{0mm}{}
\theoremstyle{specialcasestyle}

\newtheorem{lem}{Lemma}

\newtheorem{rem}{Remark}

\begin{document}

\title{Efficient Importance Sampling for Large Sums of Independent and Identically Distributed Random Variables}
\author{Nadhir Ben Rached \thanks{Chair of Mathematics for Uncertainty Quantification, Department of Mathematics, RWTH Aachen University, Aachen, Germany ({\tt benrached@uq.rwth-aachen.de}).}, Abdul-Lateef Haji-Ali \thanks{School of Mathematical \& Computer Sciences, Heriot-Watt University, Edinburgh, UK ({\tt a.hajiali@hw.ac.uk}).},  Gerardo Rubino \thanks{INRIA, Rennes - Bretagne Atlantique, France ({\tt Gerardo.Rubino@inria.fr}).}, \\ and Ra\'ul Tempone \thanks{Computer, Electrical and Mathematical Sciences \& Engineering Division (CEMSE), King Abdullah University of Science and Technology (KAUST), Thuwal, Saudi Arabia ({\tt raul.tempone@kaust.edu.sa}). Alexander von Humboldt Professor in Mathematics for Uncertainty Quantification, RWTH Aachen University, Aachen, Germany ({\tt tempone@uq.rwth-aachen.de}).}.
	\\
	%\thanks { The authors are with Computer, Electrical and Mathematical Science and Engineering (CEMSE) Division, King Abdullah University of Science and Technology (KAUST), Thuwal, Makkah Province, Saudi Arabia, (e-mail: {nadhir.benrached, abla.kammoun, slim.alouini, raul.tempone}@kaust.edu.sa).
	%}
	
		\thanks{This work was supported by the KAUST Office of Sponsored Research (OSR) under Award No. URF/1/2584-01-01 and the Alexander von Humboldt Foundation. A-L. Haji-Ali was supported by a Sabbatical Grant from the Royal Society of Edinburgh. }

}

\date{}
\maketitle
\thispagestyle{empty}
\vspace{-12mm}
\begin{abstract}
We discuss estimating the probability that the sum of nonnegative independent and identically distributed random variables falls below a given threshold, i.e., $\mathbb{P}(\sum_{i=1}^{N}{X_i} \leq \gamma)$, via importance sampling (IS).  We are particularly interested in the rare event regime when  $N$ is large and/or $\gamma$ is small. 
The exponential twisting is a popular technique for similar problems that, in most cases, compares favorably to other estimators.
 However, it has some limitations: i) it assumes the knowledge of the moment generating function of $X_i$ and ii) sampling under the new IS PDF is not straightforward and might be expensive.
 The aim of this work is to propose an alternative IS PDF that approximately yields, for certain classes of distributions and in the rare event regime,  at least the same performance as the exponential twisting technique  and,  at the same time, does not introduce serious limitations. 
The first class includes distributions whose probability density functions (PDFs) are asymptotically equivalent, as $x \rightarrow 0$, to $bx^{p}$, for $p>-1$ and $b>0$. For this class of distributions, the Gamma IS PDF with appropriately chosen parameters retrieves approximately, in the rare event regime corresponding to small values of $\gamma$ and/or large values of $N$, the  same performance of the estimator based on the use of the exponential twisting technique. In the second class, we consider the Log-normal setting, whose  PDF at zero vanishes faster than any polynomial, and we show numerically that a Gamma IS PDF with optimized parameters  clearly outperforms the exponential twisting IS PDF. Numerical experiments validate the efficiency of the proposed estimator in delivering a highly accurate estimate 
 in the regime of large $N$ and/or small $\gamma$.
 
% We apply the proposed simplified approach to a wide range of distributions and  validate numerically  that it recovers the optimal amount of variance reduction resulting from the use of the exponential twisting technique but with less computational cost.
\end{abstract}

{\bf Keywords: }Importance sampling, rare event, exponential twisting, Gamma IS PDF. 

{\bf AMS subject classifications:} 65C05, 62P30.

\section{Introduction}
Efficient estimation of rare event probabilities finds various applications in the performance evaluation/prediction of wireless communication systems operating over fading channels \cite{alouini}. In particular, the left-tail of the cumulative distribution function (CDF) of sums of  nonnegative independent and identically distributed (i.i.d.) random variables is an example of a rare event probability that is of practical importance. More specifically, the outage probability at the output of equal gain combining (EGC) and maximum ratio combining (MRC) receivers can be expressed as the CDF of the sum of fading channel envelops (for EGC) and fading channel gains (for MRC) \cite{7328688}. 

The accurate estimation of the left-tail of the CDF of sums of random variables requires the use of variance reduction techniques because the naive Monte Carlo  sampler is computationally expensive \cite{opac-b1132466,rubino2009rare,opac-b1123521}. Moreover, the existing closed-form approximations \cite{8737752,8744610,4570452,1388730,4781943,4939219,4814351} fail to be accurate when the tail of the CDF is considered. The literature is  rich in works in which variance reduction techniques were developed to efficiently estimate rare event probabilities corresponding to the left-tail of the CDF of sums of random variables, see \cite{asmussen2016exponential,7328688, 8472928,botev2019fast,gulisashvili2016,Nadhir_SLN,9014029} and the references therein.  For instance, the authors in \cite{asmussen2016exponential}  used exponential twisting, which is a popular importance sampling (IS) technique, to propose a logarithmically efficient estimator of the CDF of the sum of i.i.d. Log-normal random variables. The logarithmic efficiency is a popular property in rare event simulation  used to ensure  estimators' efficiency \cite{7328688}. Let $\hat{\alpha}$ be an unbiased estimator of $\alpha$, i.e., $\mathbb{E}[\hat{\alpha}]=\alpha$. We say that $\hat{\alpha}$ is logarithmically efficient if $\lim_{\alpha \rightarrow 0}\frac{\log(\mathbb{E}[\hat{\alpha}^2])}{\log(\alpha^2)}=1$. In \cite{gulisashvili2016}, the CDF of the sum of correlated Log-normal random variables was considered. The authors developed an IS estimator  based on shifting the mean of the corresponding  multivariate Gaussian distribution. Under mild assumptions, they proved that their proposed estimator is logarithmically efficient. Based on \cite{gulisashvili2016} and under the assumption that the left-tail sum distribution is determined by only one dominant component, the authors in \cite{Nadhir_SLN} combined IS with a control variate technique to construct an estimator with the asymptotically vanishing relative error property, which is the most desired property in the field of rare event simulations \cite{opac-b1132466}. In \cite{7328688}, two unified IS approaches were developed using the hazard rate twisting concept \cite{Juneja:2002:SHT:566392.566394,BenRached2016} to efficiently estimate the CDF of sums of independent random variables. The first estimator is shown to be logarithmically efficient, whereas the second achieves the bounded relative error property for i.i.d. sums of random variables and under the given assumption that was shown to hold for most of the practical distributions used to model the amplitude/power of fading channels. The bounded relative error is a stronger criterion than the logarithmic efficiency. We say that an unbiased estimator $\hat{\alpha}$ of $\alpha$ achieves the bounded relative error property if $\frac{\mathrm{var}[\hat{\alpha}]}{\alpha^2}$ is asymptotically bounded when $\alpha$ goes to $0$, see \cite{7328688}

The efficiency of the above mentioned estimators was studied when the number of summand $N$ was kept fixed. More specifically, recall that the objective is to efficiently estimate the probability that the sum of nonnegative i.i.d. random variables falls below a given threshold, i.e., $\mathbb{P}(\sum_{i=1}^{N}{X_i} \leq \gamma)$. A close look at the above mentioned estimators shows that the efficiency results were proved when the rarity parameter $\gamma$ decreases whereas $N$ is kept fixed. However, in most cases, the efficiency of the existing estimators is considerably affected when $N$ increases. This represents the main motivation of the present work. We aim to introduce a highly accurate estimator that efficiently estimate $\mathbb{P}(\sum_{i=1}^{N}{X_i} \leq \gamma)$ in the rare event regime when $N$ is large and/or $\gamma$ is small. 

It is well-acknowledged that the exponential twisting technique compares favorably, in most  cases,  to existing estimators. It is the optimal IS probability density function (PDF) in the sense that it minimizes the Kullback-Leibler (KL) divergence with respect to the underlying PDF under certain constraints \cite{doi:10.1177/0037549707087713}. However, it has  some limitations. First, it requires the knowledge of the moment generating function of $X_i$, $i=1,2,\cdots,N$. Second, sampling according to the new IS PDF is not straightforward and might be expensive. Moreover, the twisting parameter is not available in a closed-form expression and needs to be estimated numerically. 
Motivated by the above limitations, we summarize the main contributions of the present work as follows: 
\begin{itemize}
\item We propose an alternative IS estimator that approximately yields, for certain classes of distributions and in the rare event regime, at least the same efficiency as the one given by the estimator based on exponential twisting and at the same time does not introduce the above limitations.
\item The first class includes distributions whose PDFs vanish at zero  polynomially. For this class of distributions, the Gamma IS PDF with appropriately chosen parameters retrieves approximately, in the regime of rare events corresponding to small values of $\gamma$ and/or large values of $N$, the same performances as the exponential twisting PDF.
\item The above result does not apply to the Log-normal setting as the corresponding PDF approaches zero faster than any polynomials.  We show numerically  that in this setting, the Gamma IS PDF with optimized parameters achieves a substantial amount of  variance reduction compared to the one given by exponential twisting. 
%\item In addition to approximately yielding at least the same performance as the estimator based on exponential twisting, the estimator based on the use of the Gamma PDF as an IS PDF is easy to implement. This is compared to the estimator based on exponential twisting whose implementation is possible only under restrictive assumptions.
\item Numerical comparisons with some of the existing estimators validate that the proposed estimator can deliver highly accurate estimates with low computational cost in the rare event regime corresponding to large $N$ and/or small $\gamma.$
\end{itemize}
The paper is organized as follows. In section 2, we define the problem setting and motivate the work. In section 3, we introduce the exponential twisting approach and present its limitations. The main contribution of this work is presented in section 4, where we show that the Gamma IS PDF with optimized parameters retrieves approximately, for certain classes of distributions and in the rare event regime, at least the same performance as the exponential twisting technique. Finally, numerical experiments are shown in section 5 to compare the proposed estimator with various existing estimators.

\section{Problem Setting and Motivation}
Let $X_1,X_2,\cdots,X_N$ be i.i.d.  nonnegative random variables with common PDF $f_X(\cdot)$ and CDF $F_{X}(\cdot)$. Let $\boldsymbol{x}=(x_1,\cdots,x_N)^{t}$  and  $h_{\bold{X}}(\boldsymbol{x})=\prod_{i=1}^{N}{f_X(x_i)}$ be the joint PDF of the random vector  $(X_1,\cdots,X_N)^{t}$. We consider the estimation of 
\begin{align}\label{qoi}
\alpha(\gamma,N)=\mathbb{P}_{h_{\bold{X}}} \left (\sum_{i=1}^{N}{X_i} \leq \gamma \right ),
\end{align}
where $\mathbb{P}_{h_{\bold{X}}} (\cdot)$ is the probability under which the random vector $\bold{X}=(X_1,\cdots,X_N)^{t}$ is distributed according to $h_{\bold{X}}(\cdot)$,  i.e., for any Borel measurable set $A$ in $\mathbb{R}^N$, we have $\mathbb{P}_{h_{\bold{X}}} \left ( \bold{X} \in A\right )=\int_{A} h_{\bold{X}}(\boldsymbol{x})d\boldsymbol{x}$. 
As an application, the quantity of interest $\alpha(\gamma,N)$ could represent the outage probability at the output of EGC and MRC wireless receivers operating over fading channels. In fact, the instantaneous signal to noise ratio (SNR) at EGC or MRC diversity receivers is given as follows \cite{7328688} 
	\begin{align}\label{snr}
		\gamma_{end}=\frac{E_s}{N_0\sqrt{N^{1-p+q}}} \left ( \sum_{i=1}^{N}{R_i^p}\right )^q,
	\end{align} 
	where $N$ is the number of diversity branches, $\frac{E_s}{N_0}$ is the SNR per symbol at the transmitter, $R_i$, $i=1,2,...,N$, is the fading channel envelope and
	\begin{align}
		(p,q) = \begin{cases} (1,2) & \text{EGC}, \\ (2,1)& \text{MRC}.\end{cases}
	\end{align}
	The outage probability is defined as the probability that the SNR falls below a given threshold. Using (\ref{snr}), it can be easily shown that the outage probability at the output of EGC and MRC receivers can be expressed as the CDF of the sum of fading channel envelops (for EGC) and fading channel gains (for MRC), and hence can be expressed as in (\ref{qoi}).

We focus on the estimation of $\alpha(\gamma,N)$ when $N$ is large and/or $\gamma$ is small. Before delving into the core of the paper, we illustrate via a simple example that the efficiency of an IS estimator, that performs well when $\gamma$ decreases and $N$ is not sufficiently large, can deteriorate when we increase the values of $N$.
%\section{Approach Based on Sample Rejection }
%In this part, we present the simplest importance sampling estimator that is based on sample rejection.  As it will be illustrated later, this estimator exhibits good performances with respect to $\gamma$ but  has poor performances for large values of $N$. 
We first write the quantity of interest as
 \begin{align}\label{transformation}
 \nonumber \mathbb{P}_{h_{\bold{X}}}\left (\sum_{i=1}^{N}{X_i}\leq \gamma \right )&= \mathbb{P}_{h_{\bold{X}}}\left (\sum_{i=1}^{N}{X_i}\leq \gamma, X_i \leq \gamma \text{  }\forall i \right )\\
 \nonumber &= \mathbb{P}_{h_{\boldsymbol{w}}} \left (\sum_{i=1}^{N}{w_i}\leq 1 \right ) \left (F_{X}(\gamma)\right )^N\\
 &= \mathbb{E}_{h_{\boldsymbol{w}}} \left [ (F_X(\gamma))^N \textbf{1}_{(\sum_{i=1}^{N}{w_i} \leq 1)}\right  ]=\mathbb{E}_{h_{\boldsymbol{w}}} \left[ \hat{\alpha}(\gamma,N)\right ],
 \end{align}
where $w_i$ is equal in distribution to $\frac{X_i}{\gamma}$ conditional on the event $\{X_i \leq \gamma\}$, $i=1,2,.\cdots,N$, and $h_{\boldsymbol{w}}(\boldsymbol{w})=\prod_{i=1}^{N}{f_w(w_i)}$  with $f_{w}(\cdot)$ is the PDF of $w_i$, i.e., the conditional PDF of $\frac{X_i}{\gamma}$ given the event $\{ \frac{X_i}{\gamma} \leq 1\}$, and is given by $f_w(w)=\frac{\gamma f_X(\gamma w)}{F_X(\gamma)} \textbf{1}_{(w<\gamma)}$. Note that $\mathbb{E}_{h_{\boldsymbol{w}}}[\cdot]$ denotes the expectation under $h_{\boldsymbol{w}}(\cdot)$.
  The  estimator is then given by estimating the right-hand side term of (\ref{transformation}) by the naive Monte Carlo method
 \begin{align*}
  \hat{\alpha}_{M}(\gamma,N)=\frac{1}{M} \sum_{k=1}^{M}{(F_X(\gamma))^N \textbf{1}_{(\sum_{i=1}^{N}{w_i^{(k)}} \leq 1)}},
 \end{align*}
where $(w_1^{(k)},\cdots,w_N^{(k)})$ , $k=1,\cdots,M$, are independent realizations sampled according to $h_{\boldsymbol{w}}(\cdot)$.
Note that this estimator can be understood as applying IS with IS PDF being the truncation of the underlying PDF over the hypercube $[0,\gamma]^N$. It can be easily proved that for fixed $N$,  this estimator achieves the desired bounded relative error property with respect to the rarity parameter $\gamma$ for distributions that satisfy $f_{w}(x) \sim b x^{p}$ as $x$ approaches zero and for $p>-1$ and $b>0$, see \cite{8472928}. This property means that the squared coefficient of variation, defined as the ratio between the variance of an estimator and its squared mean, remains bounded as $\gamma \rightarrow 0$, 
  see \cite{opac-b1132466}. More precisely, when this property holds,
   the number of required samples to meet a fixed accuracy requirement remains bounded independently of how small $\alpha(\gamma,N)$ is.  The question now is what happens when $N$ is large. 
   %Let us assume for example that $X_i$ is uniform between $[0,1]$ (in this case, $w_i$ is also uniform over $[0,1]$), then
 %$$
%\mathbb{P}_{h_{\boldsymbol{w}}}(\sum_{i=1}^{N}{w_i}\leq 1)=\frac{1}{N!}.
 %$$
%Thus, the naive Monte Carlo estimator has a squared coefficient of variation  approximately equal to $N!$ which is worse than any exponential increase. In general, 
Using the Chernoff bound, we obtain for all $\eta>0$
$$
 \mathbb{P}_{h_{\boldsymbol{w}}}(\sum_{i=1}^{N}{w_i}\leq 1) \leq \exp \left ( \eta+N \log \left ( \mathbb{E}_{f_w}{}[\exp(-\eta w)]\right )\right ),
$$
where $\mathbb{E}_{f_w}[\cdot]$ denotes the expectation under $f_{w}(\cdot)$. The squared coefficient of variation of $\hat{\alpha}(\gamma,N)$ in (\ref{transformation}) is given by
\begin{align*}
\nonumber \text{SCV}(\hat{\alpha}(\gamma,N))=\frac{\mathrm{var}_{h_{\boldsymbol{w}}}[\hat{\alpha}(\gamma,N)]}{\alpha^2(\gamma,N)}&=\frac{\mathbb{P}_{h_{\boldsymbol{w}}}(\sum_{i=1}^{N}{w_i}\leq 1)(1-\mathbb{P}_{h_{\boldsymbol{w}}}(\sum_{i=1}^{N}{w_i}\leq 1))}{\left (\mathbb{P}_{h_{\boldsymbol{w}}}(\sum_{i=1}^{N}{w_i}\leq 1) \right )^2}\\
&=\frac{1-\mathbb{P}_{h_{\boldsymbol{w}}}(\sum_{i=1}^{N}{w_i}\leq 1)}{\mathbb{P}_{h_{\boldsymbol{w}}}(\sum_{i=1}^{N}{w_i}\leq 1)}
\end{align*}
In particular,  when $\eta=1$, the squared coefficient of variation (which is asymptotically equal to $1/ \mathbb{P}_{h_{\boldsymbol{w}}}(\sum_{i=1}^{N}{w_i}\leq 1) $ in the regime of rare events) is lower bounded by $\exp \left ( -1-N \log \left ( \mathbb{E}_{f_w}{}[\exp(-w)]\right )\right )$.
This shows that the squared coefficient of variation increases at least exponentially, which proves that the efficiency of the estimator deteriorates when $N$ is large. 
\section{Exponential Twisting}
In this section, we review the popular exponential twisting IS approach and enumerate its limitations in estimating the quantity of interest. When applicable, it is well-acknowledged that the exponential twisting technique is expected to produce a substantial amount of variance reduction and to compare favorably, in most  cases, to other estimators \cite{asmussen2016exponential}. For distributions with light right tails and under the i.i.d. assumption, the estimator based on exponential twisting can be proved, under some regularity assumptions, to be logarithmically efficient when the probability of interest is either $\mathbb{P}_{h_{\bold{X}}}(\sum_{i=1}^{N}{X_i}>\gamma)$ and  $\gamma \rightarrow +\infty$ or $\mathbb{P}_{h_{\bold{X}}}(\sum_{i=1}^{N}{X_i}>\gamma N)$ and  $N \rightarrow +\infty$ \cite{opac-b1123521}. In the left tail setting, which is the region of interest in the present work, the exponential twisting was shown in \cite{asmussen2016exponential} to achieve the logarithmic efficiency property in the case of i.i.d. Log-normal random variables when the probability of interest is $\mathbb{P}_{h_{\bold{X}}}(\sum_{i=1}^{N}{X_i} <N\gamma)$ and either $N\rightarrow +\infty$ or $\gamma \rightarrow 0$.

%Before that, recall that the optimal IS PDF is the truncation of the underlying PDF $h_{\bold{X}}(\cdot)$ over the rare set $\{\boldsymbol{x}, \text{  such that  }x_i \geq 0, \sum_{i=1}^{N}{x_i} \leq \gamma  \} $. While this optimal change of measure is not practical as it assumes the knowledge of the unknown parameter $\alpha(\gamma,N)$, it provides insights on how a good IS PDF should be selected in order to achieve a substantial amount of variance reduction. In fact, a good change of measure emphasizes samples to be in the rare set. Moreover, the new IS PDF should be close to the underlying PDF over the rare set, i.e., maintain the likelihood ratio constant over the rare set. 
In  \cite{doi:10.1177/0037549707087713}, the exponential twisting technique was also shown to be optimal in the sense that it minimizes the KL divergence with respect to the underlying PDF under the constraint that the rare set $\{\boldsymbol{x} \in \mathbb{R}^{N}_{+}, \text{  such that  } \sum_{i=1}^{N}{x_i} \leq \gamma  \} $ is no longer rare. The IS PDF is selected to be the solution of the following optimization problem, see \cite{doi:10.1177/0037549707087713},
\begin{align}
\nonumber &\inf_{h^*_{\bold{X}} \geq 0} \int{h^*_{\bold{X}}(\boldsymbol{x})\log \left ( \frac{h^*_{\bold{X}}(\boldsymbol{x})}{h_{\bold{X}}(\boldsymbol{x})}\right )d\boldsymbol{x}}\\
&s.t \hspace{4mm} \int{h^*_{\bold{X}}(\boldsymbol{x})d\boldsymbol{x}}=1\\
 \nonumber &\hspace{8mm} \mathbb{E}_{h^*_{\bold{X}}}\left [\sum_{i=1}^{N}{X_i} \right ]=\gamma\\
 & \nonumber \hspace{8mm} h^*_{\bold{X}}(\boldsymbol{x}) \geq 0, \hspace{2mm} x_i \geq 0 \text{ for all } i \in {1,2,\cdots,N}.
\end{align}
The solution of this problem is given as (see \cite{doi:10.1177/0037549707087713} for a more general setting)
\begin{align}
h^*_{\bold{X}}(\boldsymbol{x})=\frac{h_{\bold{X}}(\boldsymbol{x}) \exp \left (\theta^* \sum_{i=1}^{N}{x_i}\right )}{\mathbb{E}_{h_{\bold{X}}} \left [\exp \left ( \theta^* \sum_{i=1}^{N}{X_i}\right )\right ]}, \text{  }  \boldsymbol{x} \in \mathbb{R}^N_+
\end{align}
and $\theta^*$ solves
$$
\frac{\mathbb{E}_{h_{\bold{X}}}\left [\sum_{i=1}^{N}{X_i}\exp \left ( \theta^* \sum_{i=1}^{N}{X_i}\right )\right ]}{\mathbb{E}_{h_{\bold{X}}}\left [\exp \left ( \theta^* \sum_{i=1}^{N}{X_i}\right )\right ]}=\gamma.
$$
Hence,  by writing $h_{\bold{X}}^*(\boldsymbol{x})=\prod_{i=1}^{N}{f_{\bold{X}}^*(x_i)}$,  we clearly observe that the optimal density is given by exponentially twisting each univariate PDF $f_X(\cdot)$
$$
f_{X}^*(x)=\frac{f_X(x) \exp(\theta^* x)}{M(\theta^*)}, \hspace{2mm} x\geq 0,
$$
with $M(\theta)=E_{f_X}[\exp(\theta X)]$ and the optimal twisting parameter $\theta^*$ satisfies
$$
\frac{M^{'}(\theta^*)}{M(\theta^*)}=\frac{\gamma}{N}.
$$
Since the left-tail of sums of random variables is considered in this work, we have that $\theta^* \rightarrow -\infty$ as $\gamma \rightarrow 0$ and/or $N \rightarrow +\infty$  \cite{asmussen2016exponential}. Using the exponential twisting technique, the IS estimator of $\alpha(\gamma,N)$ using $M$ i.i.d. samples of $\bold{X}$ from $h_{\bold{X}}^*(\cdot)$ is given as follows
\begin{align*}
\hat{\alpha}_{\textrm{exp},M}(\gamma,N)=	\frac{1}{M}\sum_{k=1}^{M}{\textbf{1}_{(\sum_{i=1}^{M}{X_i^{(k)}} \leq \gamma)} (M(\theta^*))^N \exp \left ( -\theta^* \sum_{i=1}^{N}{X_i^{(k)}}\right )}
\end{align*}
%From the above formulation, we expect the exponential twisting technique to outperform the approaches described in  the previous subsections. 
%However, 
Observe, however, that the exponential twisting technique has some restrictive limitations. The main one is that sampling according to $f_X^*(\cdot)$ is not straightforward. One generally needs the use of an acceptance-rejection technique, the complexity of which can be dramatic when the probability of acceptance is relatively small. In such a case, the computational complexity of the algorithm can be huge and even worse than the naive Monte Carlo method. There are other less critical drawbacks. First, computations are much simpler if the moment generating function $M(\theta)$ is known in closed-form. Such a requirement does not hold in general. Also, the twisting parameter $\theta^*$ does not have, in general, a closed-form expression, and hence, it should be approximated numerically. 
\section{Gamma Family as IS PDF}
The objective of this paper is to propose an alternative IS PDF that approximately yields, for certain classes of distributions that include most of the common distributions and in the rare event regime corresponding to large $N$ and/or small $\gamma$, at least the same performance as the exponential twisting technique  and  at the same time does not introduce serious limitations. We distinguish three scenarios depending on how the PDF $f_X(\cdot)$ approaches zero.
% In what remains, the support of $f_X(\cdot)$ is assumed to be $[0,+\infty)$. We distinguish three scenarios.

\subsection{$f_X(x) \sim b$ as $x$ goes to $0$ and $b>0$ is a constant}
Recall that the exponential twisting IS PDF satisfies
$$
f_X^*(x) \propto f_X(x) \exp(\theta^* x ), \hspace{2mm} x\geq 0,
$$
with $\theta^* \rightarrow -\infty$ as $\gamma \rightarrow 0$ and/or $N \rightarrow +\infty$. Therefore, as $f(x)\sim b$ and $b>0$, and by letting $\tilde{M}(\theta)=-\frac{1}{\theta}$, we instead consider the following IS PDF
$$
\tilde{f}_X(x)=\frac{\exp(\theta x)}{\tilde{M}(\theta)},  \hspace{2mm} x\geq 0.
$$
We choose $\theta$ to be equal to $\tilde{\theta }$ such that $\frac{\tilde{M}^{'}(\tilde{\theta})}{\tilde{M}(\tilde{\theta})}=\frac{\gamma}{N}$. Through  simple computation, we obtain $\tilde{\theta}=-\frac{N}{\gamma}$. To conclude, when $f(x)\sim b$ and $b>0$, we propose an IS PDF given by the exponential distribution with rate $\frac{N}{\gamma}$. %In the regime of rare events, we expect the proposed change of measure to yield approximately  the same performance as the one given by exponential twisting.

\subsection{$f_X(x)=x^{p} g(x)$ with $g(x) \sim b$ as $x$ goes to $0$, $p>-1$, and $b>0$ is a constant}
Using the same methodology as in section 4.1, the IS PDF that we consider is
\begin{align}\label{gamma_is}
\tilde{f}_X(x)=\frac{x^{p}\exp(\theta x)}{\tilde{M}(\theta)}, \hspace{2mm} x\geq 0.
\end{align}
Therefore, the new PDF corresponds to the Gamma PDF with shape parameter $p+1$ and scale parameter $-1/\theta$. The normalizing constant is $\tilde{M}(\theta)=\frac{\Gamma(p+1)}{(-\theta)^{p+1}}$. Hence, the value $\theta$ is chosen to be equal to $\tilde{\theta}$ such that  $\frac{\tilde{M}^{'}(\tilde{\theta})}{\tilde{M}(\tilde{\theta})}=\frac{\gamma}{N}$ and is given by
\begin{align}\label{tildetheta}
\tilde{\theta}=-\frac{N}{\gamma}(p+1).
\end{align}
Using the Gamma IS PDF in (\ref{gamma_is}), the proposed IS estimator of $\alpha(\gamma,N)$ using $M$ i.i.d. samples of $\bold{X}$ from $\tilde{h}_{\bold{X}}(\boldsymbol{x})=\prod_{i=1}^{N}{\tilde{f}_{X}(x_i)}$ is
\begin{align*}
\nonumber \hat{\alpha}_{is,M}(\gamma,N)&=\frac{1}{M}\sum_{k=1}^{M}{\textbf{1}_{(\sum_{i=1}^{N}{X_i^{(k)}} \leq \gamma)}\prod_{i=1}^{N}{\frac{f_{X}(X_i^{(k)})}{\tilde{f}_X(X_i^{(k)})}}}\\
&=\frac{1}{M}\sum_{k=1}^{M}{\textbf{1}_{(\sum_{i=1}^{N}{X_i^{(k)}} \leq \gamma)}(\tilde{M}(\tilde{\theta}))^N\prod_{i=1}^{N}{\frac{f_{X}(X_i^{(k)})\exp \left (-\tilde{\theta} X_{i}^{(k)}\right )}{(X_{i}^{(k)})^p }}}
\end{align*}

\begin{figure*}
\begin{minipage}[h]{0.9775\textwidth}
\small
\begin{center}
\begin{tabular}{|l|l|l|}
\multicolumn{3}{c}{Table I: Some PDF asymptotics around zero  \footnote{ Functions $I_{\xi}(\cdot)$, and $K_{\xi}(\cdot)$ are respectively the modified Bessel functions of the first kind and order $\xi$ and the second kind and order $\xi$ \cite{gradshteyn2007}.  }}\\
\multicolumn{3}{c}{}\\
\hline
\textbf{Distribution} & PDF & Proportional to \\
&& as $x \rightarrow 0$ \\
\hline
Exponential &  $k \exp(-k x)$ & 1\\[1ex]
$k>0$  & & \\[1ex]
\hline
Gamma &  $\frac{1}{\beta^k\Gamma(k)}x^{k-1}\exp(-\frac{x}{\beta})$ & $x^{k-1}$ \\[1ex]
$k,\beta>0$  & & \\[1ex]
\hline
Weibull  & $\frac{k}{\lambda} (\frac{x}{\lambda})^{k-1} \exp(-(\frac{x}{\lambda})^k)$ & $x^{k-1}$ \\[1ex]
$k,\lambda>0$ & & \\[1ex]
\hline
Nakagami-m & $\frac{2m^m}{\Gamma(m)\Omega^m}x^{2m-1}\exp(-\frac{m}{\Omega}x^2)$  & $x^{2m-1}$ \\[1ex]
$m,\Omega>0$ & &\\[1ex]
\hline
Generalized Gamma & $\frac{p/a^d}{\Gamma(d/p)}x^{d-1}\exp (-(\frac{x}{a})^p)$  & $x^{d-1}$ \\[1ex]
$a,d,p>0$ & &\\[1ex]
\hline
Rice & $\frac{x}{\sigma^2} \exp(-\frac{x^2+\nu^2}{2\sigma^2})I_0((\frac{x\nu}{\sigma^2}))$  & $x$ \\[1ex]
 $\sigma>0,\nu \geq 0>0$ & &\\[1ex]
\hline
 Gamma-Gamma & $\frac{2(km)^{\frac{k+m}{2}}}{\Gamma(k)\Gamma(m)\Omega} (\frac{x}{\Omega})^{\frac{k+m}{2}-1} K_{k-m}\left ( 2\sqrt{\frac{km x}{\Omega}}\right )$ & $x^{k-1}$ \\[1ex]
 $\Omega>0,m>k>0,m-k \notin \mathbb{N}$ & &\\[1ex]
\hline
$\kappa-\mu$ distribution & $\frac{2\mu(1+\kappa)^{\frac{\mu+1}{2}}x^{\mu}}{\Omega^{\frac{\mu+1}{2}}\kappa^{\frac{\mu-1}{2}}\exp(\mu \kappa)}\exp(-\frac{(1+\kappa)\mu x^2}{\Omega}) I_{\mu-1} \left ( 2\mu \sqrt{\frac{\kappa(\kappa+1)}{\Omega}}x\right )$ & $x^{2\mu-1}$ \\[1ex]
$\kappa,\mu>0$ & & \\[1ex]
\hline
\end{tabular}
\end{center}
\end{minipage}
\end{figure*}

In Table I, we provide a non-exhaustive list of distributions that belong to section 4.2 (note that distributions in section 4.2 include those in section 4.1). These distributions are among the most used distributions to model the amplitudes and powers of wireless communications fading channels. 

 \begin{rem} \hspace{2mm}It is worth mentioning that for distributions satisfying $f_X(x)=x^{p} g(x)$ with $g(x) \sim b$ as $x$ goes to $0$, $p>-1$, and $b>0$ is a constant, the proposed approach with the Gamma IS PDF in (\ref{gamma_is}) with parameters $p$ and $\tilde{\theta}$ in (\ref{tildetheta})  achieves approximately, as $\gamma$ decreases to $0$ and/or $N$ increases, the same performance as the one given by the  exponential twisting without introducing serious limitations. 
Let $A_1$ and $A_2$ be the second moments of the proposed and the exponential twisting estimators, respectively. Then, the ratio between $A_1$ and $A_2$ has the following expression 
	\begin{align}\label{ratio}
	\nonumber \frac{A_1}{A_2}&=\frac{\mathbb{E}_{\tilde{h}_{\boldsymbol{x}}} \left [\bold{1}_{(\sum_{i=1}^{N}{X_i} \leq \gamma)} \prod_{i=1}^{N}{\frac{f_X^2(X_i)}{\tilde{f}_{X}^{2}(X_i)}} \right ]}{\mathbb{E}_{h^*_{\boldsymbol{x}}} \left [\bold{1}_{(\sum_{i=1}^{N}{X_i} \leq \gamma)} \prod_{i=1}^{N}{\frac{f_X^2(X_i)}{(f_{X}^{*}(X_i))^2}}\right ]}\\
		\nonumber &=\frac{(\tilde{M}(\tilde{\theta}))^{2N}\mathbb{E}_{\tilde{h}_{\boldsymbol{x}}} \left [\bold{1}_{(\sum_{i=1}^{N}{X_i} \leq \gamma)} \prod_{i=1}^{N}{g^2(X_i)} \exp \left (-2\tilde{\theta} \sum_{i=1}^{N}{X_i} \right )\right ]}{(M(\theta^*))^{2N}\mathbb{E}_{h^*_{\boldsymbol{x}}} \left [\bold{1}_{(\sum_{i=1}^{N}{X_i} \leq \gamma)} \exp \left (-2\theta^* \sum_{i=1}^{N}{X_i} \right )\right ]}\\
	\nonumber	&=\frac{(\tilde{M}(\tilde{\theta}))^N\mathbb{E}_{h_{\boldsymbol{x}}} \left [\bold{1}_{(\sum_{i=1}^{N}{X_i} \leq \gamma)} \prod_{i=1}^{N}{g(X_i)} \exp \left (-\tilde{\theta} \sum_{i=1}^{N}{X_i} \right )\right ]}{(M(\theta^*))^N\mathbb{E}_{h_{\boldsymbol{x}}} \left [\bold{1}_{(\sum_{i=1}^{N}{X_i} \leq \gamma)} \exp \left (-\theta^* \sum_{i=1}^{N}{X_i} \right )\right ]}\\
		&= \frac{(\tilde{M}(\tilde{\theta}))^N \int_{(\sum_{i=1}^{N}{x_i}  \leq \gamma)}{\prod_{i=1}^{N}{g(x_i)f_X(x_i)} \exp \left (-\tilde{\theta} \sum_{i=1}^{N}{x_i} \right )dx_1\cdots dx_N }}{(M(\theta^*))^N \int_{(\sum_{i=1}^{N}{x_i}  \leq \gamma)}{\prod_{i=1}^{N}{f_X(x_i)} \exp \left (-\theta^* \sum_{i=1}^{N}{x_i} \right )dx_1\cdots dx_N }}.
	\end{align}
First observe  that $M(\theta)=\int_{0}^{\infty}{\exp(\theta x)x^p g(x)dx}$ is well-approximated by $b\tilde{M}(\theta)=b \int_{0}^{\infty}{\exp(\theta x)x^p dx}$ for sufficiently small negative values of $\theta$. Moreover, recall that $\theta^*$ and $\tilde{\theta}$ go to $-\infty$ as either $\gamma \rightarrow 0$ or $N \rightarrow \infty$, %we deduce that $M(\theta^*)$ is well-approximated by $b \tilde{M}(\tilde{\theta})$. 
and that $\theta^*$ and $\tilde{\theta}$ satisfy $\frac{M'(\theta^*)}{M(\theta^*)}=\frac{\gamma}{N}$ and  $\frac{\tilde{M}'(\tilde{\theta})}{\tilde{M}(\tilde{\theta})}=\frac{\gamma}{N}$, respectively. Thus, as $\gamma \rightarrow 0$ and/or $N \rightarrow \infty$, we obtain that $\theta^*$ is well-approximated by $\tilde{\theta}$, and hence $M(\theta^*)$ is well-approximated by $b \tilde{M}(\tilde{\theta})$. Finally, using the latter two approximations and the fact that $g(x) \sim b$ as $x$ goes to $0$, we conclude from (\ref{ratio}) that  $A_1$ is approximately equal to $A_2$ when $\gamma$ goes to $0$. For large values of $N$, the same conclusion can be deduced  by observing that $\mathbb{E}_{f_{X}^*}[X_i]=\mathbb{E}_{\tilde{f}_X}[X_i]=\frac{\gamma}{N}$, $i=1,2,\cdots,N$. Thus, the random variables $X_1,X_2,\cdots,X_N$ take, when sampled according to the IS PDFs, sufficiently small values when $N$ is sufficiently large.
\end{rem}

 \subsection{The Log-normal Case} 
Distributions that do not approach 0 polynomially are much more difficult to handle and need to be tackled on a case-by-case basis. In this work, we consider the  case of the sum of i.i.d. standard Log-normal random variables. The density decreases to $0$ at a faster rate than any polynomials and thus the Gamma distribution with fixed shape parameter will not recover the results given by the use of the exponential twisting technique. Note that in \cite{asmussen2016exponential}, the exponential twisting technique was applied to the sum of i.i.d. standard Log-normals by i) providing an unbiased estimator of the moment generating function, ii) approximating the value of $\theta$, and iii) using acceptance-rejection to sample from the IS  PDF. 
 
 The main difficulty is that the PDF of the Log-normal distribution does not have a Taylor expansion at $x=0$. The first estimator we propose is based on truncating the support $[0,+\infty]$ and only working on $[a,+\infty]$ with $a=\delta\gamma/N$. This allows the use of a Taylor expansion at $x=a$. This procedure, however, introduces a bias that needs to be controlled. We show numerically that this estimator exhibits better performances than the one based on exponential twisting. Moreover, we observe that, in the regime of rare events, the proposed estimator achieves approximately the same performances as the Gamma IS PDF with shape parameter equal to $2$. This is the main motivation behind introducing  a second estimator whose IS PDF is a Gamma PDF with optimized parameters. The numerical results show that the second estimator achieves substantial variance reduction with respect to the first estimator.
 \subsubsection{Biased estimator}
We rewrite the quantity of interest as
\begin{align}\label{eq_ln}
\nonumber \mathbb{P}_{h_{\bold{X}}}\left ( \sum_{i=1}^{N}{X_i} \leq \gamma\right )&\approx \left (1-F_X(\frac{\delta \gamma}{N}) \right )^N\\
& \times \mathbb{P} _{h_{\bold{X}}}\left ( \sum_{i=1}^{N}{X_i}\leq \gamma \Big{|}X_i>\frac{\delta\gamma}{N}, \forall i\right ),
\end{align}
where $\delta$ is a fixed value belonging to $[0,1)$. The first factor on the right-hand side has a  known closed-form expression. Let $\bar{f}_X(\cdot)$ be the PDF of $X_i |\{X_i >\frac{\delta \gamma}{N}\}$, $i=1,2,\cdots,N$, whose expression is given as follows:
$$
\bar{f}_X(x)=\frac{1}{x\sqrt{2 \pi}} \frac{\exp \left ( -\frac{(\log(x))^2}{2}\right )}{P(X_i>\frac{\delta \gamma}{N})}, \hspace{2mm} x\geq \frac{\delta \gamma}{N}.
$$
Next, we write the second factor on the right hand side of (\ref{eq_ln}) as follows:
$$
\mathbb{P}_{h_{\bold{X}}} \left ( \sum_{i=1}^{N}{X_i}\leq \gamma \Big{|}X_i>\frac{\delta\gamma}{N}, \forall i\right )= \mathbb{P}_{\bar{h}_{\bold{X}}} \left ( \sum_{i=1}^{N}{X_i} \leq \gamma\right ),
$$
with $\bar{h}_{\bold{X}}(\boldsymbol{x})=\prod_{i=1}^{N}{\bar{f}_{X}(x_i)}$. The exponential twisting IS PDF is then given by
$$
\bar{f}_X^*(x) \propto \bar{f}_X(x) \exp(\theta x), \text{  }x \geq \frac{\delta \gamma}{N}.
$$
Now, by using the Taylor expansion of $\bar{f}_{X}(\cdot)$ at the point $x=\delta\gamma/N$, we write
$$
\bar{f}_X(x)= \bar{f}_X(\frac{\delta \gamma}{N})+(x-\frac{\delta\gamma}{N})\bar{f}_X^{'}(\frac{\delta \gamma}{N})+\frac{(x-\frac{\delta \gamma}{N})^2}{2}\bar{f}_X^{''}(\xi_{x,\delta,N}),
$$
where $\xi_{x,\delta,N}$ is between $\frac{\delta \gamma}{N}$ and $x$.
Hence, the approximate exponential twisting IS PDF is given by
\begin{align}\label{change_LN}
\tilde{f}_X(x)=\frac{ \bar{f}_X\exp ( \theta x)+(x-\frac{\delta\gamma}{N})\bar{f}_X^{'} \exp(\theta x)}{\tilde M(\theta)},\hspace{2mm} x\geq \frac{\delta \gamma}{N},
\end{align}
with the notation $\bar f_X=\bar{f}_X(\frac{\delta \gamma}{N})$ and $\bar{f}_X^{'}=\bar{f}_X^{'}(\frac{\delta \gamma}{N})$. We assume that $\frac{\delta \gamma}{N}$ is strictly less than $\exp(-1)$ to ensure that $\bar{f}_X^{'}>0$. This  assumption is not restrictive, as we are interested in the rare event regime corresponding to $N$ large and/or $\gamma$ small. Through a simple computation, we get
$$
\tilde{M}(\theta)=-\frac{\exp \left ( \theta \delta \gamma/N\right )}{\theta}\bar{f}_X+\frac{\exp \left ( \theta \delta \gamma/N\right )}{\theta^2}\bar{f}_X^{'}.
$$
The value of $\theta$ that solves $\frac{\tilde{M}^{'}(\theta)}{\tilde{M}(\theta)}=\frac{\gamma}{N}$ is given by
$$
\theta=-\frac{\bar f_X-c \bar{f}_X^{'}+\sqrt{(\bar f_X-c \bar{f}_X^{'})^2+8\bar{f}_X\bar{f}_X^{'}c}}{2c\bar f_X},
$$
with $c=\frac{\gamma}{N}(1-\delta)$. The remaining part is to sample from $\tilde{f}_X(\cdot)$. To do this, we write
$$
\tilde{f}_X(x)=-\frac{\bar{f}_X\exp(\theta \delta \gamma/N)}{\tilde{M}_X(\theta)\theta}\tilde{f}_1(x)+\frac{\bar{f}_X^{'}\exp(\theta \delta \gamma/N)}{\tilde{M}_X(\theta)\theta^2}\tilde{f}_2(x),
$$
where $\tilde{f}_1(x)=-\frac{\theta \exp(\theta x)}{\exp(\theta \delta \gamma/N)}$ and $\tilde{f}_2(x)=\frac{\theta^2(x-\delta\gamma/N)\exp(\theta x)}{\exp(\theta \delta \gamma/N)}$ are two valid PDFs for $x>\delta \gamma/N$.
%\subsubsection{Choice of the parameter $\delta$}

The question that remains is related to controlling the bias through a proper choice of the parameter $\delta$. Let $\alpha_1(\gamma,N)=\left (1-F_X(\frac{\delta \gamma}{N}) \right )^N \mathbb{P}_{h_{\bold{X}}} \left ( \sum_{i=1}^{N}{X_i}\leq \gamma \Big{|}X_i>\frac{\delta\gamma}{N}, \forall i\right )$. Then, the global relative error can be upper bounded as follows:
\begin{align}\label{error_split}
\left |\frac{\alpha(\gamma,N)-\hat{\alpha}_{1,is,M}}{\alpha(\gamma,N)} \right | \leq \frac{\alpha(\gamma,N)-\alpha_1(\gamma,N)}{\alpha(\gamma,N)}+ \left |\frac{\alpha_1(\gamma,N)-\hat{\alpha}_{1,is,M}}{\alpha_1(\gamma,N)}\right |,
\end{align}
where $\hat{\alpha}_{1,is,M}$ is the IS estimator of $\alpha_1(\gamma,N)$ based on $M$ i.i.d. realizations sampled according to $\tilde{h}_{\bold{X}}(\boldsymbol{x})=\prod_{i=1}^{N}{\tilde{f}_X(x_i)}$ where the the  PDF $\tilde{f}_X(\cdot)$ is given in (\ref{change_LN})
\begin{align*}
 \hat{\alpha}_{1,is,M}(\gamma,N)=\frac{1}{M}\sum_{k=1}^{M}{\left (1-F_X(\frac{\delta \gamma}{N}) \right )^N\textbf{1}_{(\sum_{i=1}^{N}{X_i^{(k)}} \leq \gamma)}\prod_{i=1}^{N}{\frac{\bar{f}_{X}(X_i^{(k)})}{\tilde{f}_X(X_i^{(k)})}}}.
\end{align*}
The parameter $\delta$ is then chosen to control the bias term in (\ref{error_split}), that is the first term on the right-hand side of (\ref{error_split}). The second term on the right-hand side is the statistical relative error of estimating $\alpha_1(\gamma,N)$ by $\hat{\alpha}_{1,is,M}$. From  the Central Limit Theorem (CLT), this error term is approximately proportional to the coefficient of variation of $\hat{\alpha}_{1,is,M}$. 

To achieve a global relative error of order $\epsilon$, it is sufficient to bound the two error terms, i.e., the statistical relative error and the relative bias, by $\epsilon/2$. Hence, the value of $\delta$ is selected such that the following  inequality holds
%More precisely, suppose that we aim to meet $\epsilon$ relative error, the parameter $\delta$ is then chosen such that
\begin{align}\label{bias}
0 \leq \frac{\alpha(\gamma,N)-\alpha_1(\gamma,N)}{\alpha(\gamma,N)}\leq \epsilon/2.
\end{align}
The following lemma provides the relation between $\delta$ and $\epsilon$ such that (\ref{bias}) is fulfilled.
\begin{lem}
\hspace{2mm}The following expression of $\delta(\epsilon,N,\gamma)$
\begin{align}
\delta(\epsilon,N,\gamma)=\frac{N}{\gamma} \exp \left ( \Phi^{-1} \left ( \frac{\epsilon}{2N} \frac{(\Phi(\log(\gamma/N)))^N}{(\Phi(\log(\gamma)))^{N-1}}\right )\right ),
\end{align}
where $\Phi(\cdot)$ is the CDF of the standard Normal distribution, ensures that (\ref{bias}) holds.
\end{lem}
\begin{proof}
We first write that
\begin{align}
\nonumber &\alpha(\gamma,N)-\alpha_1(\gamma,N)\\
\nonumber &=\mathbb{P}_{h_{\bold{X}}} \left ( \{\sum_{i=1}^{N}{X_i} \leq \gamma\} \cap \cup_{i=1}^{N} \{ X_i \leq \delta \gamma/N\}\right )\\
\nonumber & \leq \mathbb{P}_{h_{\bold{X}}}(\cup_{i=1}^{N} \{X_i \leq \delta \gamma/N\}\cap \cap_{i=1}^{N} \{ X_i \leq \gamma\})\\
\nonumber & \leq \sum_{i=1}^{N}\mathbb{P} _{h_{\bold{X}}}\left ( \{X_i \leq \delta \gamma/N\}\cap \cap_{j\neq i} \{ X_j \leq \gamma\}\right )\\
\nonumber &=N \mathbb{P} _{h_{\bold{X}}}\left (X_1 \leq \delta \gamma/N,X_2 \leq \gamma,\cdots, X_N \leq \gamma\right )\\
&=N \Phi(\log(\delta \gamma/N)) \left ( \Phi(\log(\gamma))\right )^{N-1}.
\end{align}
On the other hand, we have
$$
\alpha(\gamma,N) \geq \left ( \Phi(\log(\gamma/N))\right )^N.
$$
Therefore, we get
$$
\frac{\alpha(\gamma,N)-\alpha_1(\gamma,N)}{\alpha(\gamma,N)} \leq N\frac{\Phi \left (\log(\delta\gamma/N)\right ) \left ( \Phi(\log(\gamma))\right )^{N-1}}{\left (\Phi (\log(\gamma/N))\right )^{N}}.
$$
By equating the right-hand side of the above inequality with $\epsilon/2$, we obtain
$$
\delta(\epsilon,N,\gamma)=\frac{N}{\gamma} \exp \left ( \Phi^{-1} \left ( \frac{\epsilon}{2N} \frac{(\Phi(\log(\gamma/N)))^N}{(\Phi(\log(\gamma)))^{N-1}}\right )\right ),
$$ 
and hence the proof is concluded.
\end{proof}
\subsubsection{The Gamma family as an IS PDF}
When we consider a sufficiently small value of $\delta$ in the above analysis, we observe from the expression of the IS PDF in (\ref{change_LN}) that the proposed estimator with the IS PDF in (\ref{change_LN}) achieves approximately the same performance as the Gamma IS PDF with shape parameter equal to $2$. This suggests investigating whether the Gamma family can achieve further variance reduction with respect to the approach in the previous subsection. Note that the advantage of using the Gamma family as IS PDFs compared to the approach in the previous subsection is that the estimator is unbiased. Recall that the Gamma PDF is given by
\begin{align}
\tilde{f}_X(x)=\frac{x^{k-1}\exp(- x/\theta)}{\Gamma(k)\theta^k}, \text{ } x>0,
\end{align}
where $\theta>0$ and $k>0$ are the scale and shape parameters. The value of $\theta$ is chosen to be equal to $\theta=\frac{\gamma}{Nk}$ to ensure that the expected value of each of the $X_i$'s, $i=1,2,\cdots,N$, under the PDF $\tilde{f}_{X}(\cdot)$ is equal to $\frac{\gamma}{N}$. The likelihood ratio is then given by
\begin{align*}
\mathcal{L}(x_1,x_2,\cdots,x_N)=\frac{(\Gamma(k)\theta^k)^N\exp(\frac{\sum_{i=1}^{N}{x_i}}{\theta}-\frac{1}{2}\sum_{i=1}^{N}{(\log(x_i))^2})}{\prod_{i=1}^{N}{x_i^k}(\sqrt{2 \pi})^N}.
\end{align*}
The second moment of the IS estimator is bounded by
\begin{align*}
& \mathbb{E}_{\tilde{h}_{\bold{X}}} \left [\mathcal{L}^2(X_1,X_2,\cdots,X_N) \bold{1}_{(\sum_{i-1}^{N}{X_i} \leq \gamma)} \right ]\\
& \leq (\frac{\Gamma(k)\theta^k}{\sqrt{2\pi}})^{2N} \exp(\frac{2\gamma}{\theta})\\
& \times \mathbb{E}_{\tilde{h}_{\bold{X}}} \left [  \exp(-\sum_{i=1}^{N}{(\log(x_i))^2}-2k \sum_{i=1}^{N}{\log(x_i)}) \right ]\\
& \leq (\frac{\Gamma(k)(\frac{\gamma}{Nk})^k}{\sqrt{2\pi}})^{2N} \exp(2kN+k^2N).
\end{align*}
The  last  upper bound is found by maximizing the function $x \rightarrow -(\log(x))^2-2k \log(x)$ for $x>0$.
Next, using Stirling's formula for the gamma function
$\Gamma(k)= \sqrt{2\pi}k^{k-\frac{1}{2}} \exp(-k)(1+\mathcal{O}(\frac{1}{k}))$, we get
\begin{align*}
&\mathbb{E}_{\tilde{h}_{\bold{X}}} \left [\mathcal{L}^2(X_1,X_2,\cdots,X_N) \bold{1}_{(\sum_{i-1}^{N}{X_i} \leq \gamma)} \right ]\\
& \lesssim C k^{-N} \left ( \frac{\gamma}{N}\right )^{2Nk} \exp (k^2N)\\
&=C\exp(N(k^2-2k\log(N/\gamma)-\log(k)))
\end{align*}
where $C$ is a constant.  Next, the value of $k$ is chosen such that it minimizes the above right-hand side term.  The solution of this minimization problem is given as follows:
\begin{align}\label{optimal_k_value}
k^*=\frac{1}{2} \left (\log(\frac{N}{\gamma})+\sqrt{(\log(\frac{N}{\gamma}))^2+2} \right ).
\end{align}
Note that when $N$ is large and/or $\gamma$ is small, the value of $k^*$ satisfies $k^* \sim \log(\frac{N}{\gamma})$.
\section{Numerical results}
In this section, we show some selected numerical results to compare the performance of the proposed estimators compared to some of the existing estimators. We consider three scenarios depending on the distribution of $X_i$, $i=1,2,\cdots,N$: the Weibull, the Gamma-Gamma, and the Log-normal distributions. Note that the proposed approach is not restricted to these three distributions (see Table I  for a non-exhaustive list of distributions that can be handled). 

We recall that the squared coefficient of variation of an unbiased estimator $\hat{\alpha}(\gamma,N)$ of $\alpha(\gamma,N)$ has the following expression
\begin{align}
\text{SCV}(\hat{\alpha}(\gamma,N))=\frac{\mathrm{var} \left [ \hat{\alpha}(\gamma,N)\right ]}{\alpha^2(\gamma,N)}.
\end{align}
Note that, from the CLT, the number of required samples  to meet $\epsilon$ statistical relative error with $95\%$ confidence is equal to $(1.96)^2 \text{SCV}(\hat{\alpha}(\gamma,N))/\epsilon^2$. Therefore, when we compare two estimators, the one with the smaller squared coefficient of variation exhibits better performance than the other.
\subsection{Weibull Case}
In this section, we assume that $X_i$, $i=1,2,\cdots,N$, are distributed according to the Weibull distribution whose PDF is given in Table I.
%$$
%f_{X}(x)=\frac{k}{\lambda} (\frac{x}{\lambda})^{k-1} \exp \left ( -(x/\lambda)^k\right ), \text{  }x>0.
%$$
The comparison is made with respect to the second IS approach of \cite{7328688} that is based on using the hazard rate twisting (HRT).  In Figure \ref{fig_weibull_1} and Figure \ref{fig_weibull_2}, we plot the squared coefficient of variations given by  the HRT technique and the proposed approach for two different values of the shape parameter: $k=1.5$ and $k=0.5$, respectively. 
The value of $\alpha(\gamma,N)$ ranges approximately from $10^{-20}$ to $10^{-6}$ (respectively from $10^{-16}$ to $10^{-6}$) using the system's parameters of Figure \ref{fig_weibull_1} (respectively of Figure \ref{fig_weibull_2}).
These figures show that the proposed approach clearly outperforms the one based on HRT.
For instance, when $k=1.5$, $\lambda=1$, $\gamma=0.5$, and $N=12$, the proposed approach is approximately $270$ times more efficient than the one based on HRT. More specifically, to meet the same accuracy, the number of samples needed by the approach based on HRT should be approximately $270$ times the number of samples needed by the proposed approach.

\begin{figure}[h!]
\centering
   \includegraphics[scale=0.45]{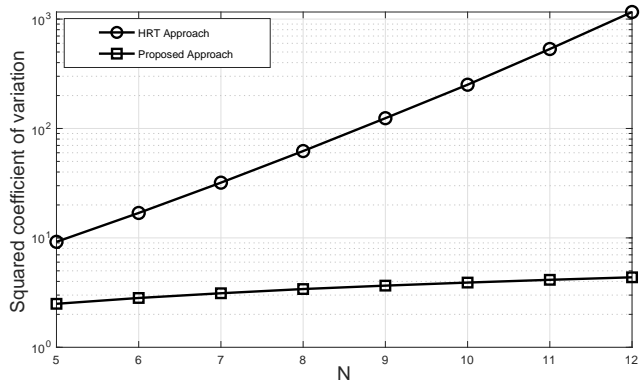}
   \caption{Squared coefficient of variation as a function of $N$ where $X_i$ are i.i.d. Weibull random variables with rate $\lambda=1$, $k=1.5$, and $\gamma=0.5$.}
   \label{fig_weibull_1}
   \end{figure}

   \begin{figure}[h!]
\centering
   \includegraphics[scale=0.45]{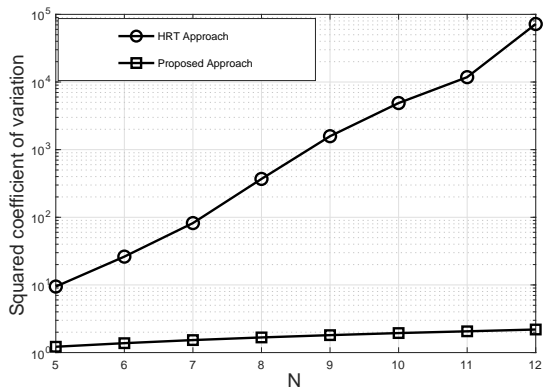}
   \caption{Squared coefficient of variation as a function of $N$ where $X_i$ are i.i.d. Weibull random variables with rate $\lambda=1$, $k=0.5$, and $\gamma=0.01$.}
   \label{fig_weibull_2}
   \end{figure}
% The outperformance of the proposed approach compared to the one based on maximum entropy is very clear. This result is expected since the maximum entropy approach can be shown to have infinite variance for values of $k$ that are less than $0.5$. This explains the oscillatory behavior of the squared coefficient of variation resulting from the maximum entropy estimator.

In the next experiment, we aim to compare the proposed approach with the HRT one when $N$ is fixed and $\gamma$ decreases. 
%In Figure \ref{fig_weibull_3}, we plot the estimated values of $\alpha(\gamma,N)$ as a function of $\gamma$ for two different values of $N$ ($N=8$ and $N=10$) given by the HRT and the proposed approaches. We observe from this figure that both approaches yield accurate estimates of $\alpha(\gamma,N)$ using $M=2\times 10^6$ samples.
%\begin{figure}[h!]
%	\centering
%	\includegraphics[scale=0.50]{./figures/alpha_vs_gamma}
%	\caption{$\alpha(\gamma,N)$ as a function of $\gamma$ where $X_i$ are i.i.d Weibull RVs with rate $\lambda=1$, $k=1.5$, and $M=2\times 10 ^6$.}
%	\label{fig_weibull_3}
%\end{figure}
In Figure \ref{fig_weibull_4}, we compare the efficiency of both approaches in terms of squared coefficient of variations plotted as a function of $\gamma$ for two scenarios depending on the value of $N$ ($N=8$ and $N=10$). In this case, the value of $\alpha(\gamma,N)$ ranges approximately from $10^{-16}$ to $10^{-6}$ for $N=8$ and from $10^{-22}$ to $10^{-8}$ for $N=10$.
We observe a clear outperformance of the proposed approach compared to the one based on using HRT for both values of $N$. While the HRT approach was proved in \cite{7328688} to achieve the bounded relative error property with respect to $\gamma$ and for a fixed value of $N$, it is clear from Figure \ref{fig_weibull_4} that the asymptotic bound increases substantially with respect to $N$, and hence the performance of the HRT approach is dramatically affected by increasing $N$. On the other hand, we observe that increasing the value of $N$ has a minor effect on the efficiency of the proposed approach, i.e., the squared coefficient of variation is approximately unchanged for both values of $N$ and for the considered range of $\gamma$. 
\begin{figure}[h!]
	\centering
	\includegraphics[scale=0.50]{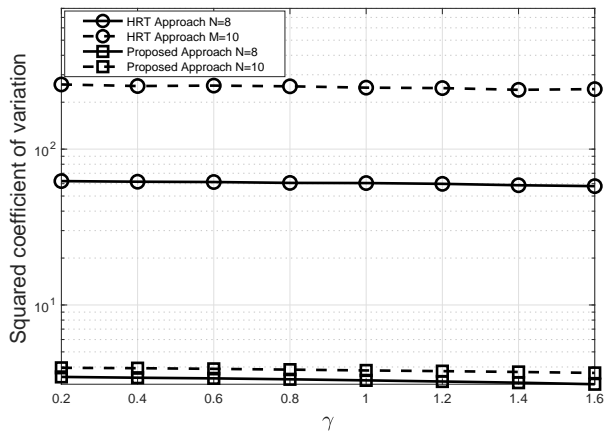}
	\caption{Squared coefficient of variation as a function of $\gamma$ where $X_i$ are i.i.d. Weibull random variables with rate $\lambda=1$, $k=1.5$.}
	\label{fig_weibull_4}
\end{figure}
This numerical observation suggests to conclude that the proposed approach satisfies the bounded relative error property with an asymptotic bound that increases with a very slow rate, compared to the one given by the HRT approach,  as we increase $N$. For illustration, the proposed approach is  approximately $18$ (respectively $64$) times more efficient than the HRT one when $N=8$ (respectively $N=10$) and $\gamma=0.2$. Note that the previous observations are valid independently of the value of $\alpha(\gamma,N)$ (see Figure \ref{fig_weibull_4}, where the squared coefficient of variation is approximately constant for a fixed value of $N$ and for the considered range of $\gamma$). This experiment and the numerical results in Figures \ref{fig_weibull_1} and \ref{fig_weibull_2} validate the ability of the  proposed approach to deliver a very accurate and efficient estimate of $\alpha(\gamma,N)$ when $N$ increases and/or $\gamma$ decreases.

\subsection{Gamma-Gamma Case}
The Gamma-Gamma distribution is used for various  challenging applications in wireless communications. For instance, it exhibited a good fit to experimental data and was used to model wireless radio-frequency channels \cite{EP} and to model atmospheric turbulences in free-space optical communication systems \cite{8238201}. The PDF of $X_i$ is given in Table I. 
 \begin{figure}[h!]
\centering
   \includegraphics[scale=0.5]{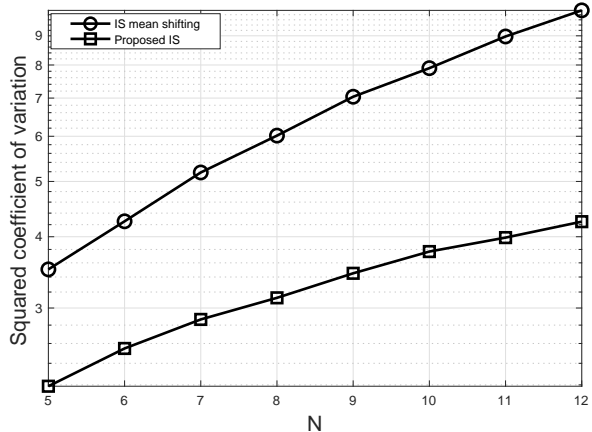}
   \caption{Squared coefficient of variation as a function of $N$ where $X_i$ are i.i.d. Gamma-Gamma random variables with $m=4$, $k=1.7$, $\Omega=1$, and $\gamma=0.5$.}
   \label{fig_GG}
   \end{figure}
In Figure~\ref{fig_GG}, we compare the proposed approach with the one in \cite{7835220} by plotting the corresponding squared coefficient of variations as a function of $N$ and for a fixed value of $\gamma$. Note that in \cite{7835220}, the proposed IS PDF is simply another Gamma-Gamma PDF with shifted mean. We call this method the IS-based mean-shifted approach. The range of the quantity of interest $\alpha(\gamma,N)$ is approximately from $10^{-18}$ to $10^{-5}$. We observe that the proposed estimator outperforms the one in \cite{7835220}. Also, we observe that the outperformance of the proposed estimator compared to the one based on mean shifting increases as we increase $N$. Moreover, we should note here that the cost per sample (in terms of CPU time) of the approach in \cite{7835220} is twice the cost of the proposed approach. This is because a Gamma-Gamma random variable is generated by the product of two independent Gamma random variables, see \cite{5425871}. For illustration, we observe from Figure~\ref{fig_GG} that when $N=12$, the proposed approach is approximately $2.5$ times (five times if we include the computing time in the comparison) more efficient than the one of \cite{7835220}.

\subsection{Log-normal Case}
The Log-normal distribution can be used to model several types of attenuation including shadowing \cite{580779},  and  weak-to-moderate turbulence channels in free-space optical communications \cite{6966082}. The standard Log-normal PDF (the associated Gaussian random variable has zero mean and unit variance) is given by
$$
f_X(x)=\frac{1}{x\sqrt{2\pi}} \exp \left ( -\frac{(\log(x))^2}{2}\right ), \text{  }x>0.
$$

   \begin{figure}[h!]
\centering
   \includegraphics[scale=0.55]{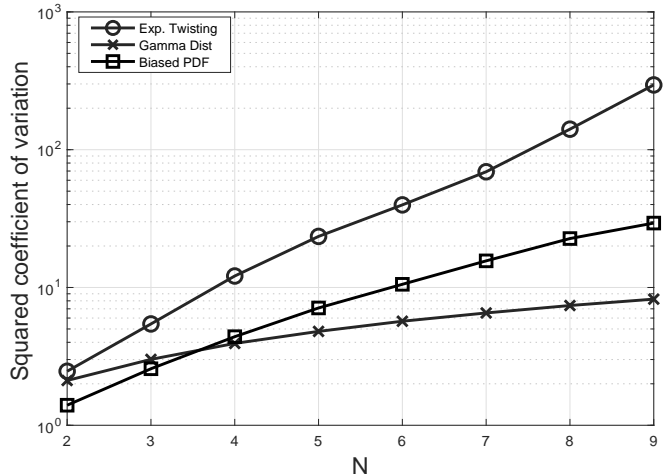}
   \caption{Squared coefficient of variation as a function of $N$ where $X_i$ are i.i.d. standard Log-normal random variables with $\gamma=0.5$, and $\epsilon=0.05$.}
   \label{fig_ln_1}
   \end{figure}
Figure \ref{fig_ln_1} shows the squared coefficient of variation given by  the exponential twisting  \cite{asmussen2016exponential}, and the two proposed approaches, i.e., the one based on the biased estimator and the other based on using the Gamma distribution as an IS PDF. The value of $\alpha(\gamma,N)$ ranges  approximately from $10^{-20}$ to $10^{-2}$.
For the considered range of $N$, we observe that out of these three approaches, it is the one using the Gamma distribution as an IS PDF that outperforms the others. When $N=9$, it is approximately $30$ times more efficient than the one based on exponential twisting. In addition to the efficiency in terms of number of samples, it is worth recalling that the exponential twisting technique developed in \cite{asmussen2016exponential} is computationally expensive in terms of computing time compared to the proposed approaches. Moreover, Figure \ref{fig_ln_1} also shows that the approach based on the biased estimator achieves better performances than the one based on exponential twisting. It is important to mention here that,  for the comparison to be fair,  the required number of samples of the biased estimator should be multiplied by $4$. This follows from the error analysis in (\ref{error_split}), in which the statistical relative error should be bounded by $\epsilon/2$, where $\epsilon$ is the required global relative error.

In Figure \ref{fig_ln_2}, we plot the squared coefficient of variations given by the three approaches as a function of $\gamma$ and for two different values of $N$ ($N=8$ and $N=10$). The quantity of interest $\alpha(\gamma,N)$ ranges approximately from $10^{-15}$ to $10^{-6}$ for $N=8$ and from $10^{-21}$ to $10^{-9}$ for $N=10$.
We observe that the approach based on using the Gamma distribution as an IS PDF clearly  asymptotically outperforms the two other approaches. For both values of $N$, the outperformance increases as we decrease $\gamma$. Moreover, the biased estimator exhibits better performances than the exponential twisting one for both values of $N$ and for the considered range of $\gamma $ values.
Furthermore, increasing $N$ has a considerable negative effect on the performances of the exponential twisting and the biased IS-based approaches.  On the other hand, Figure \ref{fig_ln_2} shows that increasing $N$ does not largely effect the performance of the IS estimator based on the use of the Gamma distribution as an IS PDF. For illustration, the approach based on using the Gamma distribution as an IS PDF is approximately $15$ times (respectively $35$) more efficient that the exponential twisting one when $N=8$ (respectively $N=10$) and $\gamma=0.6$.
   
   \begin{figure}[h!]
\centering
   \includegraphics[scale=0.5]{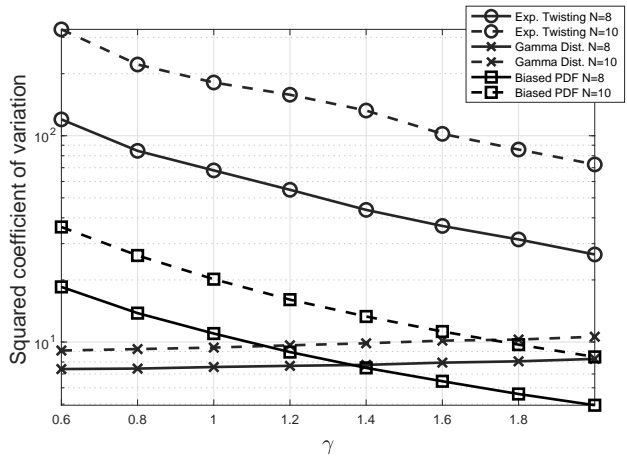}
   \caption{Squared coefficient of variation as a function of $\gamma$ where $X_i$ are i.i.d. standard Log-normal random variables with  $\epsilon=0.05$.}
   \label{fig_ln_2}
   \end{figure}

 It is important to mention that the outperformance of the estimator based on using the Gamma distribution as an IS PDF over the one based on using the biased estimator is expected. As it was mentioned above, the latter approach gives approximately the same performance as the Gamma distribution with shape parameter equal to $2$ while the former one uses the Gamma distribution as an IS PDF with an optimized shape parameter (the shape parameter was chosen to minimize an upper bound of the second moment of the proposed estimator, see the expression of $k^*$ in (\ref{optimal_k_value})).

All of the above comparisons have been carried out in terms of the number of sampled needed to meet a fixed accuracy requirement. In order to include the computing time in our comparison, we define the Work Normalized Relative Variance (WNRV) metric of an unbiased estimator $\hat{\alpha}(\gamma,N)$ of $\alpha(\gamma,N)$ as follows (see \cite{8472928}):
\begin{align}
\text{WNRV}(\hat{\alpha}(\gamma,N))=\frac{\text{SCV}(\hat{\alpha}(\gamma,N))}{M} \times \text{computing time in seconds}.
\end{align}
The computing time  is the time in seconds needed to get an estimator of $\alpha(\gamma,N)$ using $M$ i.i.d. samples of $\hat{\alpha}(\gamma,N)$. When comparing two estimators, the one that exhibits less WNRV is more efficient than the other estimator. More precisely, an estimator is efficient in terms of WNRV than another estimator means that it achieves less relative error for a given computational budget, or equivalently it needs less computing time to achieve a fixed relative error. Using the same setting as in Figure \ref{fig_ln_2}, we plot in Figure \ref{fig_ln_3} the WNRV metric as a function of $\gamma$  for two scenarios depending on the value of $N$ ($N=8$ and $N=10$). 
\begin{figure}[h!]
\centering
   \includegraphics[scale=0.5]{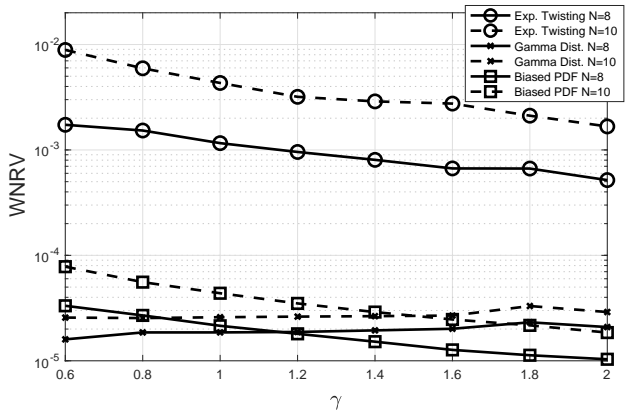}
   \caption{WNRV as a function of $\gamma$ where $X_i$ are i.i.d. standard Log-normal random variables with  $\epsilon=0.05$.}
   \label{fig_ln_3}
   \end{figure}
We observe that as $\alpha(\gamma,N)$ is getting smaller, it is the approach based on using the Gamma PDF as an IS PDF that outperforms the two other approaches in terms of WNRV (the efficiency increases as the event becomes rarer). It is worth recalling that the WNRV of the approach based on biased PDF should be multiplied by $4$ in order for the analysis to be fair (this follows from the error analysis that was performed in section 4.3.1). Moreover, Figure \ref{fig_ln_3} shows that, in addition to reducing the variance, as shown in Figure~\ref{fig_ln_2}, the approach based on using the Gamma IS PDF also reduces the computing time compared to the one using the exponential twisting technique. To see that, for $N=10$ and $\gamma=0.6$, the approach based on using the Gamma IS PDF is approximately $35$ times (respectively $340$ times) more efficient than the one based on exponential twisting when using the squared coefficient of variation metric (respectively the WNRV metric). More specifically, the Gamma based IS approach approximately reduces the computing time by a factor of $10$ with respect to the exponential twisting approach.

\section{Conclusion}
We developed efficient importance sampling estimators to estimate the rare event probabilities corresponding to the left-tail of the cumulative distribution  function of large sums of nonnegative independent and identically distributed random variables. The proposed estimators achieve asymptotically at least the same performance as the exponential twisting technique, in the regime of rare events and for certain classes of distributions that include most of the common distributions. The main conclusion is that the Gamma PDF with suitably chosen parameters achieves for most of the common distributions substantial variance reduction, and at the same time avoids the restrictive limitations of the exponential twisting technique. The numerical results validate the efficiency of the proposed approach in being able to accurately and efficiently estimate the quantity of interest in the rare event regime corresponding to large $N$ and/or small $\gamma$. One possible extension of the present work is to connect it to the works in \cite{BESKOS20171417,Jasra} by creating a sequence of approximate measures corresponding to increasing the values of $N$.

\bibliography{References}
\bibliographystyle{plain}
\end{document}